\documentclass[twoside,leqno,11pt]{article}
\usepackage[algoruled,linesnumbered,algo2e,vlined]{algorithm2e}
\usepackage{amsmath}
\usepackage{times}
\usepackage{pifont}
\usepackage{amsmath,amssymb}
\usepackage{url}

\usepackage{rotating}
\usepackage{ltexpprt} 
\usepackage{authblk} 

\newcommand{\Prefix}{\mathit{prefix}}
\newcommand{\lcp}{\mathit{lcp}}

\newcommand{\LPF}{\mathit{LPF}}
\newcommand{\PrevOcc}{\mathit{PrevOcc}}

\newcommand{\PSV}{\mathit{PSV}}
\newcommand{\NSV}{\mathit{NSV}}

\newcommand{\suf}[1]{\mathit{suf}({#1})}
\newcommand{\EMPTY}{\mathit{EMPTY}}

\newcommand{\SA}{\mathit{SA}}
\newcommand{\RANK}{\mathit{SA}^{-1}}
\newcommand{\ignore}[1]{}
\newtheorem{definition}{Definition}

\newcommand{\Lbkts}{\mathit{Lbkts}}
\newcommand{\Lbkte}{\mathit{Lbkte}}
\newcommand{\Sbkts}{\mathit{Sbkts}}
\newcommand{\Sbkte}{\mathit{Sbkte}}
\newcommand{\Array}{\mathit{A}}
\newcommand{\LMSSA}{\mathit{LMS\_SA}}
\newcommand{\Phia}{\mathit{\Phi}}

\newcommand{\prestep}{{\rm preliminary step}}
\newcommand{\parstep}{{\rm parsing step}}

\title{
  Space Efficient Linear Time Lempel-Ziv Factorization on Constant~Size~Alphabets
}

\author[1, 2]{Keisuke Goto}
\author[1]{Hideo Bannai}
\affil[1]{Department of Informatics, Kyushu University, Japan}
\affil[2]{Japan Society for the Promotion of Science (JSPS)}
\affil[ ]{\textit {\{keisuke.gotou, bannai\}@inf.kyushu-u.ac.jp}}

\date{}
\begin{document}
\maketitle
\begin{abstract}
  We present a new algorithm for computing the Lempel-Ziv Factorization (LZ77)
  of a given string of length $N$ in linear time,
  that utilizes only $N\log N + O(1)$ bits of working space, i.e., 
  a single integer array, for constant size integer alphabets.
  This greatly improves the previous best space requirement for 
  linear time LZ77 factorization (K\"arkk\"ainen et al. CPM 2013),
  which requires two integer arrays of length $N$.
  Computational experiments show that despite the added complexity of the algorithm,
  the speed of the algorithm is only around twice as slow as previous fastest 
  linear time algorithms.
\end{abstract}

\section{Introduction}
Lempel-Ziv (LZ77) factorization~\cite{LZ77} is one of the most important concepts in string
processing with countless applications in compression~\cite{LZ77,rytter03:_applic_lempel_ziv},
as well as efficient string processing~\cite{kolpakov99:_findin_maxim_repet_in_word,duval04:_linear}.
More recently, its importance has been reasserted due to the highly repetitive 
characteristics of modern datasets, such as collections of genome sequences,
for which compression schemes based on LZ77 have been shown to be particularly 
effective~\cite{kreft11:_self_based_lz77}.
Thus, time and space efficient computation of LZ77 factorization is a very important and 
heavily studied topic~(See~\cite{a.ss:_lempel_ziv_lz77} for a survey). 

In this paper, we focus on worst case linear time algorithms for computing the LZ77
factorization of a given text.
All existing linear time algorithms are based on the suffix array,
which can be constructed in linear time independent of alphabet size, when assuming
an integer alphabet.
The earlier algorithms further compute and utilize several other auxiliary integer arrays 
of length $N$, such as the inverse suffix array,
the longest common prefix (LCP) array~\cite{Kasai01},
and the Longest Previous Factor (LPF) array~\cite{crochemore08:_comput_longes_previous_factor},
and thus until recently, required at least 3 auxiliary integer arrays of length $N$ in
addition to the text.
Since all values in the LCP and LPF arrays are not required for computing the LZ factorization,
the most efficient recent linear time algorithms~\cite{goto13_lz,karkkainen_lz}
avoid constructing these arrays altogether.

The currently fastest linear time LZ-factorization algorithm,
as well as the currently most space economical linear time LZ-factorization algorithm,
have been proposed by K\"{a}rkk\"{a}inen et al.~\cite{karkkainen_lz}
They proposed 3 algorithms KKP3, KKP2, and KKP1,
which respectively store and utilize 3, 2, and 1 auxiliary integer arrays of 
length $N$ kept in main memory.
All three algorithms compute the LZ-factorization of the input text
given the text and its suffix array.
KKP3 is very similar to LZ\_BG~\cite{goto13_lz}, 
but is modified so that array accesss are more cache friendly,
thus making the algorithm run faster.
KKP2 is based on KKP3, but further reduces one integer array 
by an elegant technique that rewrites values on the integer array.
KKP1 is the same as KKP2, except that it assumes that the suffix array
is stored on disk, but since the values of the suffix array are only accessed sequentially,
the suffix array is streamed from the disk.
Thus, KKP1 can be regarded as requiring only a single integer array to be held in memory.
In this sense, KKP1 is the most space economical linear time algorithm, 
and has been shown to be faster than KKP2, if we assume that the suffix array is
already computed and exists on disk~\cite{karkkainen_lz}.
However, note that the {\em total} space requirement of KKP1 is 
still two integer arrays, one existing in memory and the other existing on disk.



In this paper, we propose new algorithms for computing the LZ77 factorization
that uses only a single auxiliary integer array of length $N$.
We achieve this by introducing a series of techniques for
rewriting the various auxiliary integer arrays from one to another,
in linear time and in-place, i.e., using only constant extra space.
Computational experiments show that our algorithm is at most 
around twice as slow as previous algorithms, but in turn, uses only half the total space,
and may be a viable alternative when the total space (including disk)
is a limiting factor due to the enormous size of data.


\section{Preliminaries}
Let $\Sigma$ be a finite {\em alphabet}.
In this paper, we assume that $\Sigma$ is an integer alphabet of constant size.
An element of $\Sigma^*$ is called a {\em string}.
The length of a string $T$ is denoted by $|T|$.
The empty string $\varepsilon$ is the string of length 0,
namely, $|\varepsilon| = 0$.
Let $\Sigma^+ = \Sigma^* - \{\varepsilon\}$.
For a string $S = XYZ$, $X$, $Y$ and $Z$ are called
a \emph{prefix}, \emph{substring}, and \emph{suffix} of $T$, respectively.
The set of prefixes of $T$ is denoted by $\Prefix(T)$.
The \emph{longest common prefix} of strings $X,Y$, denoted $\lcp(X, Y)$, 
is the longest string in $\Prefix(X) \cap \Prefix(Y)$.

The $i$-th character of a string $T$ is denoted by 
$T[i]$ for $1 \leq i \leq |T|$,
and the substring of a string $T$ that begins at position $i$ and
ends at position $j$ is denoted by $T[i..j]$ for $1 \leq i \leq j \leq |T|$.
For convenience, let $T[i..j] = \varepsilon$ if $j < i$,
$\suf{i}$ indicates $T[i..|T+1|]$,
and $T[|T|+1] = \$$ where $\$$ is a special delimiter character that
does not occur elsewhere in the string.

\subsection{Suffix Arrays}
The suffix array~\cite{manber93:_suffix} $\SA$ of any string $T$
is an array of length $|T|$ such that
for any $1 \leq i \leq |T|$,
$\SA[i] = j$ indicates that $\suf{j}$ is the $i$-th lexicographically smallest suffix of $T$.
For convenience, we assume that $\SA[0]=\SA[N+1]=0$.
The inverse suffix array $\RANK$ of $\SA$ is an array of length $|T|$ such that
$\RANK[\SA[i]] = i$.
As in~\cite{karkkainen09_plcp}, let $\Phia$ be an array of length $|T|$ such that
$\Phia[\SA[1]] = |T|$ and
$\Phia[\SA[i]] = \SA[i-1]$ for $2 \leq i \leq |T|$,
i.e., for any suffix $j = \SA[i]$,
$\Phia[j] = \SA[i-1]$ is the immediately preceding suffix in
the suffix array.
The suffix array $\SA$ for any string of length $|T|$
can be constructed in $O(|T|)$ 
time regardless of the alphabet size, assuming an integer
alphabet~(e.g.~\cite{Karkkainen_Sanders_icalp03, nong11_sais}).
Furthermore, there exists a linear time suffix array construction algorithm
for a constant alphabet using $O(1)$ working space~\cite{nong13_sacak}.

\subsection{LZ Encodings}
LZ encodings are dynamic dictionary based encodings with many variants.
The variant we consider is also known as the s-factorization~\cite{crochemore84:_linear}.

\begin{definition}[LZ77-factorization]
  \label{def:s_factorization}
  The s-factorization of a string $T$ is
  the factorization $T = f_1 \cdots f_n$ where each
  s-factor $f_k\in\Sigma^+~(k=1,\ldots,n)$ starting at
  position $i=|f_1\cdots f_{k-1}|+1$ in $T$
  is defined as follows:
  If $T[i] = c \in \Sigma$ does not occur before $i$
  then $f_k = c$. Otherwise, $f_k$ is the longest prefix of $\suf{i}$
  that occurs at least once before $i$.
\end{definition}
Note that each LZ factor can be represented in constant space,
i.e., 
a pair of integers where the first and second elements
respectively represent the length and position of 
a previous occurrence of the factor.
If the factor is a new character and the length of its previous
occurrence is $0$, the second element will encode the new character
instead of the position.
For example the s-factorization of the string
  $T = \mathtt{abaabababaaaaabbabab}$ is
  $\mathtt{a}$,
  $\mathtt{b}$,
  $\mathtt{a}$,
  $\mathtt{aba}$,
  $\mathtt{baba}$,
  $\mathtt{aaaa}$,
  $\mathtt{b}$,
  $\mathtt{babab}$. This can be represented as
  $(0, \mathtt{a})$, $(0, \mathtt{b})$, $(1,1)$, $(3,1)$,
  $(4,5)$, $(4,10)$, $(1,2)$, $(5,5)$.

We define two functions $\LPF$ and $\PrevOcc$ below.
For any $1\leq i \leq N$,
$\LPF(i)$ is the longest length of longest common prefix
between $\suf{i}$ and $\suf{j}$ for any $1 \leq j < i$, and
$\PrevOcc(i)$ is a position $j$ which gives
$\LPF(i)$\footnote{There can be multiple choices of $j$, but
  here, it suffices to fix one.}.
More precisely,
\begin{eqnarray*}
\LPF(i) &=& \max(\{0\}\cup\{\lcp(\suf{i},\suf{j}) \mid 1 \leq j < i \})\\
\mbox{and}\\
\PrevOcc(i) &=& 
\begin{cases}
  -1      & \mbox{if } \LPF(i) = 0\\
  j       & \mbox{otherwise}
\end{cases}
\end{eqnarray*}
where $j$ satisfies $1 \leq j < i$, and $T[i:i+\LPF(i)-1] = T[j:j+\LPF(i)-1]$.
Let $p_k = |f_1 \cdots f_{k-1}|+1$.
Then, $f_k$ can be represented as a pair
$(\LPF(p_k), \PrevOcc(p_k))$ if $\LPF(p_k)>0$,
and $(0, T[p_k])$ otherwise.

Crochemore and Ilie~\cite{crochemore08:_comput_longes_previous_factor}
showed that candidates values for $PrevOcc(i)$
can be reduced to only 2 position, namely,
the previous smaller value (PSV) and the next smaller value (NSV)~\cite{crochemore08:_comput_longes_previous_factor},
which are defined as follows:
\begin{eqnarray*}
  \PSV[i] &=& \SA[j_1] \\
  \NSV[i] &=& \SA[j_2]
\end{eqnarray*}
where 
$j_1 = \max(\{0\}\cup \{ 1 \leq j < \RANK[i] \mid \SA[j] < \SA[i]\})$
and 
$j_2 = \min(\{N+1\}\cup \{ N \geq j > \RANK[i] \mid \SA[j] < \SA[i] \})$.

In what follows, 
we assume that the algorithms output each LZ factor sequentially,
and will not include the total size of the LZ factorization
in the working space.


\section{Previous Algorithm}
We first describe the 3 variants (KKP3, KKP2, and KKP1) of the LZ factorization algorithm
proposed by K\"{a}rkk\"{a}inen et.al~\cite{karkkainen_lz}.
KKP3 consists of two steps, which we shall call the {\prestep} and the \parstep.
In the {\prestep}, KKP3 computes $\PSV$ and $\NSV$ for all positions and stores them in integer arrays.
Although we defer the details,
the $\PSV$ and $\NSV$ arrays can be computed in linear time by sequentially scanning $\SA$ of $T$,
and is based on the peak elimination by Crochemore and Ilie~\cite{crochemore08:_comput_longes_previous_factor}.
Then, in the \parstep,
KKP3 computes the LZ-factorization by a naive comparison between
$\suf{i}$ and $\suf{\PSV[i]}$, as well as
$\suf{i}$ and  $\suf{\NSV[i]}$,
for all positions $i$ that a factor starts 
(See Algorithm~\ref{algo:pnsv_from_sa} in Appendix.
$\lcp(i, j)$ computes the length of the longest prefix between
$\suf{i}$ and $\suf{j}$ in $O(\lcp(i,j))$ time).
In order to compute a factor $f_j$, the algorithm compares
at most twice $|f_j|$ characters.
Since the sum of the length of all the factors is $N$,
the {\parstep} of the algorithm runs in linear time.
KKP3 needs 3 integer arrays, $\SA$, $\PSV$ and $\NSV$ arrays in the {\prestep},
and 2 integer arrays $\PSV$ and $\NSV$ in the {\parstep}.
Therefore KKP3 runs in linear time using
a total of 3 auxiliary integer arrays ($\SA,\PSV,\NSV$) of length $N$.

For KKP2, K\"{a}rkk\"{a}inen et al. show that the {\parstep} can be accomplished
by using only the $\NSV$ array.
The idea is based on a very interesting connection between
$\PSV$, $\NSV$, and $\Phia$ arrays.
They showed that starting from the $\NSV$ array,
it is possible to sequentially scan and rewrite the $\NSV$ array (consequently to the $\Phia$ array)
in-place,
during which, values of $\PSV$ (and naturally $\NSV$) for each position can be 
obtained sequentially as well.
\begin{lemma}[\cite{karkkainen_lz}]\label{lem:nsv2phi}
  Given the $\NSV$ array of a string $T$ of length $N$,
  $\PSV(i)$ and $\NSV(i)$ of $T$ can be sequentially obtained
  for all positions $i=1,\ldots, N$
  in $O(N)$ total time using 
  $O(1)$ space other than the $\NSV$ array and $T$.
\end{lemma}
By making use of this technique, only the $\NSV$ array is now required for the {\parstep}.
KKP2 uses 2 integer arrays ($\SA$ and $\NSV$) in the {\prestep},
and 1 integer array ($\NSV$) in the {\parstep}, and thus in summary,
KKP2 runs in linear time using a total of 2 auxiliary integer arrays of length $N$.

We can see that the memory bottleneck of KKP2 is in the {\prestep}, 
i.e., the computation of the $\NSV$ array, where 
the space for $\SA$ is required as input, and the space for $\NSV$ is required for output.
This is because 
elements of $\SA$ are in lexicographic order and elements of $\NSV$ are in text order.
Although the scanning on $\SA$ can be sequential,
the writing to $\NSV$ is not, and both arrays must exist simultaneously.
KKP1 partly overcomes this problem, by first storing $\SA$ to disk,
and then streams the $\SA$ from the disk,
storing only the $\NSV$ array in main memory.
Thus,
KKP1 runs in linear time keeping only 1 auxiliary integer
array of length $N$ in {\em main} memory,
although of course, 
the total storage requirement is still 2 integer arrays ($\SA$ and $\NSV$).

\section{New Algorithm using a single integer array}
In this section, we describe our linear time LZ77 factorization algorithm 
that uses only a single auxiliary integer array of length $N$.
As described in the previous section, once the $\NSV$ array has been obtained,
the {\parstep} can be performed within the time and space requirements
due to Lemma~\ref{lem:nsv2phi}.
What remains is how to compute $\NSV$ using only a single integer array,
including the $\NSV$ array itself.

Our algorithm achieves this in two steps. 
We first show in Section~\ref{subsec:nsvtophi} that,
given the $\Phia$ array, $\NSV$ can be computed in linear time and $O(1)$ extra space,
by rewriting $\Phia$ array in-place.
Then, we show in Section~\ref{subsec:inplacephi} that,
given $T$, the $\Phia$ array can be computed in linear time and $O(1)$ extra space.
By combining the two algorithms, we obtain our main result.
\begin{theorem}
  Assuming a constant size integer alphabet,
  the LZ77 factorization of a string of length $N$ can be computed in $O(N)$ time
  using of $N\log N + O(1)$ bits of total working space, i.e.,
  a single auxiliary integer array of length $N$.
\end{theorem}

We call the algorithm that uses two integer arrays by incorporating the former technique, BGtwo, and the algorithm that uses only a single integer array by incorporating both techniques, BGone.
(See Figure~\ref{fig:flow_chart})

\subsection{In-place computation of the $\NSV$ array from the $\Phia$ array}\label{subsec:nsvtophi}
Since $\Phia[i]$ for each $i$ 
indicates lexicographic predecessor of $\suf{i}$,
we can sequentially access values of $\SA$ from right to left,
by accessing the $\Phia$ starting from the lexicographically largest suffix,
which is $\Phia[0]$.
More precisely, since the $\SA$ is a permutation of the integers $1,\ldots,N$,
$\Phia$ can be regarded as an array based implementation of
a singly linked list, linking the elements of $\SA$ from right to left.
Thus, the algorithm for computing $\NSV$ from $\SA$ can be simulated using the
$\Phia$ array.
An important difference is that while elements of $\SA$ are in lexicographic order,
elements of $\Phia$ are in text order, which is the same as $\NSV$.
Also, since the access on $\SA$ is sequential, the value $\Phia[i]$
is not required anymore after it is processed,
and we can rewrite $\Phia[i]$ to $\NSV[i]$.
The pseudo code of the algorithm is shown in Algorithm~\ref{algo:nsv_from_phi}.
The correctness and running time follows from the above arguments.

\begin{lemma}
  \label{lem:phi2nsv}
  Given the $\Phia$ array of a string $T$,
  $\NSV$ array of $T$ can be computed from $\Phia$ in linear time and in-place
  using $O(1)$ working space.
\end{lemma}

\begin{algorithm2e}[t]
  \caption{In-place computation of $\NSV$ from $\Phia$.}
  \label{algo:nsv_from_phi}
  \SetKwInOut{Input}{Input}\SetKwInOut{Output}{Output}
  \Input{$\Phia$ array (denoted as $\NSV$)}
  $cur \leftarrow \NSV[0]$ \tcp*[l]{$\Phia[0]$: lexicographically largest suffix}
  $prev \leftarrow 0$ \;
  \While{$cur \neq 0$}{
    \While{$cur < prev$}{
      $prev \leftarrow \NSV[prev]$ \tcp*[l]{peak elimination}
    }
    $next \leftarrow \NSV[cur]$ \tcp*[l]{$\Phia[cur]$}

    $\NSV[cur] \leftarrow prev$ \;
    $prev \leftarrow cur$ \;
    $cur \leftarrow next$ \;
  }
\end{algorithm2e}

\begin{figure}[h]
  \centerline{\includegraphics[width=0.8\textwidth]{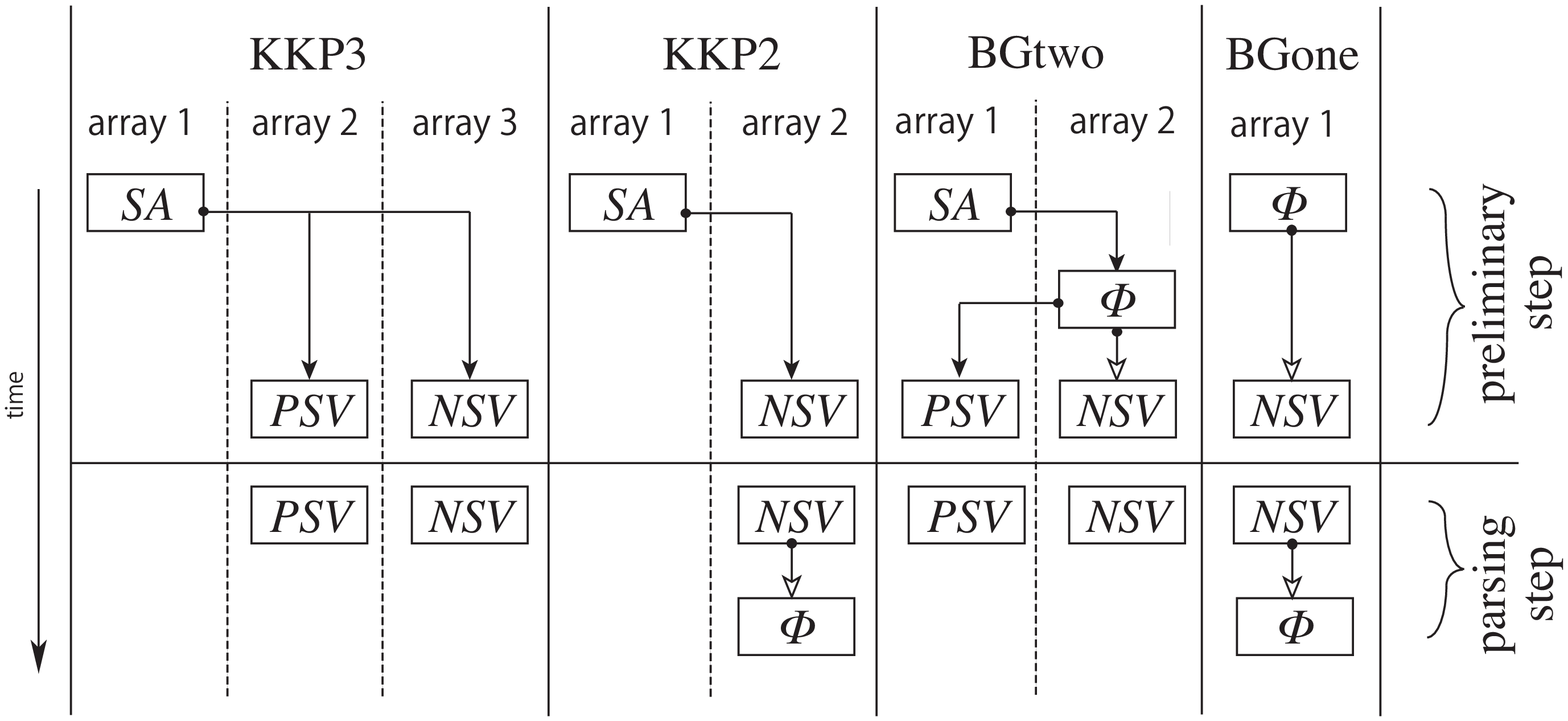}}
  \caption{
    A comparison of the auxiliary arrays used and how their contents
    change with time for the KKP variants and our algorithm.
  }
  \label{fig:flow_chart}
\end{figure}

\subsection{Computing the $\Phia$ array using $O(1)$ working space}
\label{subsec:inplacephi}
In the previous section, we showed that the $\NSV$ array
can be computed from the $\Phia$ array in-place in linear time.
By combining Lemma~\ref{lem:nsv2phi} and Lemma~\ref{lem:phi2nsv},
if the $\Phia$ array is given,
we can compute the LZ-factorization in linear time
by rewriting $\Phia$ array to $\NSV$ array in-place,
and rewriting $\NSV$ array to $\Phia$ array in-place (and sequentially obtain $\NSV$ and $\PSV$ values),
using only constant extra working space.
The problem is now how to compute the $\Phia$ array.
Although the $\Phia$ array can easily be computed in linear time by a naive sequential scan on $\SA$,
storage for both the input $\SA$ and output $\Phia$ array is required for such an approach,
as in the case of computing $\NSV$ from $\SA$.
As far as we know, an in-place linear time construction algorithm for the $\Phia$ 
array has not yet been proposed.
Below, we propose the first such algorithm.

As noted in the previous subsection, the $\Phia$ array can be
considered as an alternative representation of $\SA$, which allows
us to simulate a sequential scan on the $\SA$. Thus, in order to
construct $\Phia$ in-place, our algorithm simulates the in-place 
suffix array construction algorithm by Nong~\cite{nong13_sacak}
which runs in linear time on constant size integer alphabets.
We first describe the outline of the algorithm by Nong for computing $\SA$,
and then describe how to modify this to compute the $\Phia$ array.

\subsubsection{Construction of the suffix array by induced sorting~\cite{nong13_sacak}}
\label{subsec:nong}
Nong's algorithm is based on induced sorting,
which is a well known technique for linear time suffix sorting.
Induced sorting algorithms first sort a certain subset of suffixes,
either directly or recursively, and then induces the lexicographic order
of the remaining suffixes by using the lexicographic order of the subset.
There exist several methods depending on which subset of suffixes to choose. 
Nong's algorithm utilizes the concept of LMS suffixes defined below.
\begin{definition}
  For $1 \leq i \leq N$,
  a suffix $\suf{i}$ is an L-suffix if $\suf{i}$ is lexicographically larger than $\suf{i+1}$, 
  and an S-suffix otherwise.
  We call S or L the {\em type} of the suffix.
  An S-suffix $\suf{i}$ is a Left-Most-S-suffix (LMS-suffix)
  if $\suf{i}$ is an S-suffix and $\suf{i-1}$ is an L-suffix.
\end{definition}
Recall that $T[N+1]=\$$, where $\$$ is a special delimiter character that
does not occur elsewhere in the string.
We define $\suf{N+1}$ to be an S-suffix.
Notice that for $i \leq N$,
$\suf{i}$ is an S-suffix iff $T[i]<T[i+1]$,
or $T[i]=T[i+1]$ and $\suf{i+1}$ is an S-suffix.
The type of each suffix can be determined by scanning $T$ from right
to left.

In $\SA$, all suffixes starting with the same character $c$ occur
consecutively, and we call the interval on the suffix array of such
suffixes, the $c$-interval.
A simple observation is that the L-suffixes that start with
some character $c$ must be lexicographically smaller than
all S-suffixes that start with the same character $c$.
Thus a $c$-interval can be partitioned into to two sub-intervals,
which we call the L-interval and S-interval of $c$.

The induced sorting algorithm consists of the following steps.
We denote the working array to be $\SA$, which will become the 
suffix array of the text at the end of the algorithm.
\begin{enumerate}
\item 
  Sort the LMS-suffixes.\\
  We call the result $\LMSSA$.
  \label{sais:lms_sort}
  We omit details of how this is computed,
  since our algorithm will use the algorithm described in~\cite{nong13_sacak} as is,
  but it may be performed in linear time using $O(1)$ extra working space.
  We assume that the result $\LMSSA$ is stored in the first
  $k$ elements of $\SA$, i.e. $\SA[1..k]$,
  where $k$ is the number of LMS-suffixes.  
\item
  Put each LMS-suffix into the S-interval of its first character,
  in the same order as $\LMSSA$. \label{sais:putlms}\\
  We scan $T$ from right to left, and for each $c\in\Sigma$,
  compute and store the number of L-suffixes and S-suffixes, that start with $c$.
  We also compute the number of suffixes that start with a character that
  is lexicographically smaller than $c$.
  Storing these values requires only constant space,
  since we assume a constant size alphabet.
  From these values, we can determine the start and end positions
  of the L-interval and S-interval for any $c$. 
  Initially, all intervals are marked empty.
  By also maintaining a pointer to the left-most or right-most empty 
  element in an interval, adding elements to an L-interval or S-interval
  can also be performed in $O(1)$ time using $O(1)$ extra space.
  By a right to left scan on $\LMSSA$ (i.e. $\SA[1..k]$),
  we put each LMS-suffix in the right most empty element
  of the S-interval of the corresponding character

\item
  Sort and put the L-suffixes in their proper positions in $\SA$.\\ 
  This is done by scanning $\SA$ from left to right.
  For each position $i$, if $\SA[i] > 1$ and $\suf{\SA[i]-1}$ is an L-suffix, 
  $\suf{\SA[i]-1}$ is put in the left-most empty position of the L-interval for
  character $T[\SA[i]-1]$. 
  The correctness of the algorithm follows from the fact that 
  if suffix $\suf{\SA[i]-1}$ is an L-suffix, then, 
  $\suf{\SA[i]}$ must have been located before $i$ (in the correct order), in $\SA$.

  \label{sais:ltype_sort}
  
\item Sort and put the S-suffixes in their proper positions in $\SA$.\\
  This is done by scanning $\SA$ from right to left.
  For a position $i$, if $SA[i] > 1$ and $\suf{SA[i]-1}$ is an S-suffix,
  $\suf{SA[i]-1}$ is put in the right most empty position of the S-interval for character $T[SA[i]-1]$.
  \label{sais:stype_sort}
  The correctness of the algorithm follows from the fact that 
  if suffix $\suf{\SA[i]-1}$ is an S-suffix, then, 
  $\suf{\SA[i]}$ must have been located after $i$ (in the correct position), in $\SA$.

\end{enumerate}

In total, the algorithm computes suffix array in linear time
using only a single integer array and constant extra working space.
Note that for any position $i$, determining whether suffix $\suf{\SA[i]-1}$ 
is an L-suffix or not, can be done in $O(1)$ time using no extra space.
If $T[\SA[i]-1] < T[\SA[i]]$ then it is an S-suffix, and
if $T[\SA[i]-1] > T[\SA[i]]$ then it is an L-suffix.
For the case of $T[\SA[i]-1] = T[\SA[i]]$, the type of 
$\suf{\SA[i]-1}$ is the same as that of $\suf{\SA[i]}$,
which can be determined by the position $i$, and the start and end positions
of the L and S-intervals of character $T[\SA[i]]$.

\subsubsection{Construction of the $\Phia$ array by induced sorting}
\label{sec:const_phi}
We regard $\Phia$ as an array based implementation of
a singly linked list containing elements of $\SA$ from right to left.
The basic idea of our algorithm to construct the $\Phia$ array is to
modify Nong's algorithm for computing $\SA$, to use this list representation instead.
However, there are some technicalities that need to be addressed.

We denote the working array to be $\Array$,
which will be an array based representation of a singly linked list that links
(in lexicographic order) the set of so-far sorted suffixes at each step,
and will become the $\Phia$ array of the text at the end of the algorithm.
The algorithm is described below.

\begin{enumerate}
\item 
  Sort the LMS-suffixes.\\
  First, we sort LMS-suffixes in the same way as~\cite{nong13_sacak}.
  The result will be called $\LMSSA$ and stored in $\Array[1..k]$,
  where $k$ is the number of LMS-suffixes.

\item  \label{phiis:arrange_lms}
  Put each LMS-suffix into the S-interval of its first character,
  in the same order as $\LMSSA$. \label{phiis:putlms}\\
  In this step, we transform $\LMSSA$ to the array based linked list representation,
  so that
  for each LMS-suffix $\suf{\LMSSA[i]}$, its
  lexicographically succeeding LMS-suffix $\suf{\LMSSA[i+1]}$ 
  will be put in $\Array[\LMSSA[i]]$,
  i.e., $\Array[\LMSSA[i]]=\LMSSA[i+1]$ for $i < k$.
  If $\LMSSA$ and $\Array$ were different arrays, then we could simply
  set $\Array[\LMSSA[i]] = \LMSSA[i+1]$ for each $i < k$.
  The problem here is that since $\LMSSA$ is stored in $\Array[1..k]$,
  when setting a value at some position of $\Array$, we may overwrite 
  a value of $\LMSSA$ which has not been used yet. We overcome this problem as follows.
  
  First, we memorize $\LMSSA[1]$, the first value of $\LMSSA$.
  Then, for $1 \leq i \leq k$,
  we set $\Array[2i]=\LMSSA[i]$ and $\Array[2i-1] = \EMPTY$
  by scanning $\Array[1..k]$ from right to left.
  Since $k$ never exceeds $N/2$, we have $2i \leq N$ for all $1\leq i\leq k$.

  Next, for $1 \leq i \leq k-1$,
  let $j_1 = \Array[2i] (= \LMSSA[i])$ and $j_2 = \Array[2(i+1)] (= \LMSSA[i+1])$.
  We attempt to set $\Array[j_1] = j_2$ .
  If $\Array[j_1]=\EMPTY$, then we simply set $\Array[j_1]=j_2$.
  Otherwise $j_1 = 2i^\prime$ for some $1\leq i^\prime \leq k$,
  and $\Array[j_1]$ stores the value $\LMSSA[i^\prime]$.
  Therefore, we do not overwrite this value, but instead,
  borrow the space immediately preceding position $j_1$, and
  set $\Array[j_1-1]=j_2$.
  An important observation is that $\Array[j_1-1]$ must have been $\EMPTY$,
  because LMS-suffixes cannot, by definition, start at consecutive positions,
  and if $j_1$ was an LMS suffix, $j_1-1$ cannot be an LMS suffix and the
  algorithm will never try to set another value at this position.

  After this, we set $\Array[2i]=\EMPTY$ for all $1 \leq i \leq k$,
  and we arrange the remaining values to their correct positions by
  attempting to traverse succeeding suffixes stored in $\Array$ from
  the lexicographically smallest suffix of
  $\LMSSA$ memorized at the beginning of the process.
  Let $i$ be the current position we are traversing.
  We attempt to obtain its succeeding suffix by reading $\Array[i]$.
  If $\Array[i] \neq \EMPTY$, the succeeding suffix of $\suf{i}$ was stored at correct position,
  and we continue with the next position $\Array[i]$.
  If $\Array[i] = \EMPTY$, then the succeeding suffix of $\suf{i}$ may be stored at
  the immediately preceding position, i.e. $\Array[i-1]$.
  In such a case, $\Array[i-1] \neq \EMPTY$, and
  we set $\Array[i]=\Array[i-1]$ and $\Array[i-1]=\EMPTY$,
  and continue with the next position $\Array[i]$.
  If $\Array[i-1] = \EMPTY$, this means that $\suf{i}$ is the lexicographically 
  largest suffix of LMS-suffixes, and we finish the process.

  In this way, for all LMS-suffixes $\suf{i}$,
  we can set the succeeding suffix at $\Array[i]$.
  The process essentially scans the values of $\LMSSA$ on $\Array$ twice.
  Therefore, this step runs in $O(k)$ time and $O(1)$ working space.
\item
  \label{phiis:ltype_sort}
  Sort and put the L-suffixes in their proper positions in $\Array$.\\ 
  To simulate the algorithm for $\SA$,
  we need to scan the suffixes in lexicographically increasing order by
  using $\Array$.
  Let $\suf{i}$ be a suffix the algorithm is processing.
  We want to set $\Array[j]=i-1$ if $\suf{i-1}$ is an L-suffix,
  and $\suf{j}$ is the suffix that lexicographically precedes suffix $\suf{i-1}$.

  To accomplish this, we introduce four integer arrays of size  $|\Sigma|$ each,
  $\Lbkts[c]$, $\Lbkte[c]$, $\Sbkts[c]$ and $\Sbkte[c]$.
  $\Lbkts[c]$ and $\Lbkte[c]$
  store the lexicographically smallest and largest suffix
  of the L-interval for a character $c$ which have been inserted into
  $\Array$,
  and $\Sbkts[c]$ and $\Sbkte[c]$ are the same for each S-interval.
  All values are initially set to $\EMPTY$.
  We first scan the list of LMS suffixes in lexicographically increasing 
  order represented in $\Array$ constructed in the previous step,
  and insert each LMS suffixes into the corresponding S-interval,
  by updating $\Sbkts[c]$ and $\Sbkts[e]$.
  Then, we scan all LMS- and L-suffixes in lexicographically increasing order
  by traversing the succeeding suffixes on $\Array$ by
  starting from $\Lbkts[c]$, 
  traversing the list represented by $\Array$ 
  until we process $\Lbkte[c]$.
  Then we do the same starting from $\Sbkts[c]$ and process the 
  suffixes until we reach $\Sbkte[c]$, 
  and repeat the process for all character $c$ in lexicographic order.  
  
  Let $\suf{i}$ be a suffix the algorithm is currently processing.
  We store $\suf{i-1}$ in the appropriate position of $\Array$,
  if $\suf{i-1}$ is an L-suffix, and do nothing otherwise.
  Since we know the type of suffix $\suf{i}$ since we are either processing
  a suffix between $\Lbkts[c]$ and $\Lbkte[c]$ or $\Sbkts[c]$ and $\Sbkte[c]$,
  the type of $\suf{i-1}$ can be determined in constant time by simply comparing
  $T[i-1]$ and $T[i]$, i.e. it is an L-suffix if $T[i-1] > T[i]$, 
  an S-suffix if $T[i-1] < T[i]$, and has the same type as $\suf{i}$
  if $T[i-1] = T[i]$.

  When storing $\suf{i-1}$ in $\Array$, we check $\Lbkts[T[i]]$.
  If $\Lbkts[T[i-1]]=\EMPTY$, then, $\suf{i-1}$ is the lexicographically smallest
  suffix starting with $T[i-1]$. 
  We set $\Lbkts[T[i-1]]=\Lbkte[T[i-1]]=i-1$.
  Otherwise, there is at least one suffix lexicographically smaller than
  $\suf{i-1}$ in the L-interval for character $T[i-1]$.
  This suffix is $\Lbkte[T[i-1]] = j$, and
  we set $\Array[j]=i-1$, and update $\Lbkte[T[i-1]]=i-1$.

  In this way we can compute all the lexicographically succeeding suffix
  of each L-suffixes in the corresponding L-interval,
  and store them in $\Array$.
  Since the number of times we read the values of $\Array$
  is at most the number of LMS- and L-suffixes,
  and the updates for each new L-suffix can be done in $O(1)$ time,
  the algorithm runs in linear time using only a single integer array
  and $O(1)$ working space in total.
  
\item Sort and put the S-suffixes in their proper positions in $\Array$.\\
  To simulate the algorithm for $\SA$,
  we need to scan all L-suffixes in lexicographically 
  decreasing order by using $\Array$.
  However, since the linked list of L-suffixes constructed on $\Array$ 
  in the previous step
  is in increasing order, we first rewrite $\Array$ to reverse the direction of 
  the links.
  That is, we want to set $\Array[j]=i-1$ if $\suf{i-1}$ is an L-suffix
  and $\suf{j}$ is the suffix that lexicographically succeeds suffix
  $\suf{i-1}$.

  This rewriting can be done by scanning the succeeding suffixes in a
  similar way as that of Step~\ref{phiis:ltype_sort}.
  For each $c$ in lexicographically increasing order, traverse the L-suffixes by
  using $\Lbkts[c], \Lbkte[c]$, and $\Array$, and simply
  rewrite the values in $\Array$ to reverse the links, i.e.,
  if $\suf{j}$ preceded $\suf{i}$ then $\Array[i] = j$.
  
  Now we have a lexicographically decreasing list of L-suffixes represented in $\Array$,
  and want to insert the S-suffixes into $\Array$.
  The process is similar to that of Step~\ref{phiis:ltype_sort}.
  Initially the values for $\Sbkts[c]$ and $\Sbkte[c]$ for all $c$ are set to $\EMPTY$.
  Then, for each $c$ in lexicographically {\em decreasing} order,
  we traverse preceding suffixes on $\Array$ by
  starting from $\Sbkte[c]$, 
  traversing the list represented by $\Array$ 
  until we process $\Sbkts[c]$.
  Then we do the same starting from $\Lbkte[c]$ and process the 
  suffixes until we reach $\Lbkts[c]$, and so on.
  Let $\suf{i}$ be a suffix the algorithm is currently processing.
  If $\suf{i-1}$ is an S-suffix, we store $\suf{i-1}$ in the appropriate
  position of $\Array$ and update $\Sbkts[c]$ and $\Sbkte[c]$ accordingly,
  and do nothing otherwise.
  A minor detail during this process is that we also link preceding
  suffixes which are in different S or L intervals.  

  Now that all suffixes have been inserted and linked, we can obtain all suffixes
  in decreasing order by traversing preceding suffixes on $\Array$, 
  i.e. $\Array$ is now equal to the $\Phia$ array.
  Similarly to the previous step, we can see that this step runs in linear time using
  one integer array of length $N$ ($\Array$) and $O(1)$ extra space.

\end{enumerate}
All steps run in linear time using $\Array$ and $O(1)$ extra space,
thus giving a linear time algorithm for computing $\Phia$ using $O(1)$ extra working space.

The above procedure describes how to construct $\Phi$ from $T$ using only a 
single integer array of length $N$.
We propose another variant of the algorithm that, given $\SA$,
computes the $\Phi$ by rewriting $\SA$ in-place
in linear time and $O(1)$ extra working space.
The idea may seem useless at a glance, but may have applications
when the $\SA$ is already available, since the conversion does not
require the expensive recursion step as in the linear time $\SA$ construction 
algorithm (in Step 1), but can be achieved in a few scans.

\begin{lemma}
  Given the $\SA$ of a string $T$ of length $N$,
  $\Phia$ array of $T$ can be computed from $\SA$ in $O(N)$ time and in-place
  using $O(1)$ working space.
\end{lemma}
\begin{proof}
  It suffices to compute $\LMSSA$, since then we can run the above algorithm
  from Step~\ref{phiis:arrange_lms}.
  We scan $T$ from right to left, and for each character $c$,
  count the number of L- and S-suffixes that start with $c$,
  and obtain the L- and S-interval for each character $c$ on $\SA$.
  Let $k$ be a counter of the number of LMS suffixes initially set to $0$.
  We then scan $\SA$ from left to right for $1 \leq i \leq N$.
  If $i$ is within an S-interval and $T[\SA[i]] < T[\SA[i]-1]$,
  then, $\suf{\SA[i]}$ is an LMS-suffix and
  we store it in $\SA[k+1]$, and increment $k$.

  In this way, we can obtain $\LMSSA$ and also $\SA$
  by applying Step~\ref{sais:putlms}-\ref{sais:stype_sort}
  in $O(N)$ time and
  $O(1)$ extra working space.
\end{proof}

\subsection{In-place computation of $\SA$ from the $\Phia$ array}
An advantage of the KKP algorithms compared to BGone may be that 
$\SA$ is left untouched after the LZ-factorization.
On the other hand, the $\Phia$ array is left after running BGone.
Actually, it is possible to show that the $\Phia$ array can be converted back to $\SA$ in 
linear time and in-place, using $O(1)$ extra working space.
The proof of the following lemma is given in the Appendix.
\begin{lemma}
  \label{lem:phi2sa}
  Given a string $T$ and its $\Phia$ array,
  the $\SA$ array of $T$ can be computed
  in linear time and in-place using $O(1)$ working space.
\end{lemma}


\section{Computational Experiments}
We implemented BGtwo and two variations of BGone,
these are differ in the computation of $\Phia$ array.
One of which computes $\Phia$ array directly from $T$ (BGoneT),
and the other firstly computes $\SA$ and then computes
$\Phia$ array from $\SA$ (BGoneSA).
The 3 implementation are available at \url{http://code.google.com/p/bgone/}.
We compared our algorithms
with the implementation of KKP1, KKP2, and KKP3
\footnote{\url{https://www.cs.helsinki.fi/group/pads/lz77.html}.}.
We use SACA-K which is the implementation of Nong's algorithm to compute $\LMSSA$ in BGoneT,
and use divsufsort to compute $\SA$ in the other implementations,
BGtwo, BGoneSA, KKP1, KKP2, and KKP3.
Note that in terms of speed, BGoneT has a disadvantage
since although divsufsort is not a truly linear time algorithm, 
it is generally faster than SACA-K.
These conditions were chosen since the latter algorithms can choose any 
suffix array construction algorithm, while BGoneT cannot.

All computations were conducted on a Mac Xserve (Early 2009)
with 2 x 2.93GHz Quad Core Xeon processors and 24GB Memory,
only utilizing a single process/thread at once.
The programs were compiled using the GNU C++ compiler ({\tt g++}) 4.7.1
with the {\tt -Ofast -msse4.2} option for optimization.
The running times are measured in seconds, starting from after 
reading input text in memory, and the average of 3 runs is reported.
We use the data used in previous work
\footnote{\url{http://pizzachili.dcc.uchile.cl/texts.html}, \url{http://pizzachili.dcc.uchile.cl/repcorpus.html}.}.
Table~\ref{table:time} shows running times of the algorithms,
and how many integer arrays is used.

The results show that the runtimes of our algorithms is only 
about twice as slow as KKP1, despite the added complexity introduced
so that the algorithm can run on a single integer array.
One reason that KKP1 is faster may be because BGone 
needs random access on the integer array to compute the $\NSV$ array,
while KKP1 does not.
Although KKP1 needs to write and read $\SA$ to and from the disk,
sequential I/O seems to be faster than random access on the memory.
BGoneSA which computes $\Phia$ array through $\SA$ is a little faster than
BGoneT which computes $\Phia$ directly.

\begin{table}[t]
  \label{table:time}
  \caption{
    Time and space consumption for computing LZ factorization.
    The times were measured after reading input text in memory.
    The runtime of KKP1 includes the writing and reading time of $\SA$
    to and from the disk.
  }
  \begin{center}
    

\newcommand{\wid}{0.8cm}
\newcommand{\widb}{1.3cm}
\begin{tabular}{|l|r|r|r|r|r|r|r|r|}
\hline
Algorithm & KKP1 & KKP2 & KKP3 & BGtwo & BGoneT & BGoneSA \\ \hline
\# of Arrays & 2 & 2 & 3 & 2 & 1 & 1 \\ \hline
\hline

proteins.200MB & 75.57  & 67.50  & 56.41  & 80.30  & 157.85  & 136.13 \\ \hline
english.200MB & 69.34  & 61.11  & 49.84  & 75.50  & 156.25  & 132.29 \\ \hline
dna.200MB & 72.35  & 64.01  & 52.76  & 79.46  & 146.46  & 136.85 \\ \hline
sources.200MB & 55.58  & 47.45  & 38.53  & 57.68  & 116.18  & 98.76 \\ \hline
coreutils & 54.83  & 46.68  & 37.45  & 58.64  & 116.94  & 101.34 \\ \hline
cere & 140.82  & 122.83  & 107.51  & 177.94  & 323.00  & 299.18 \\ \hline
kernel & 69.19  & 59.04  & 49.81  & 77.15  & 154.16  & 131.48 \\ \hline
einstein.en.txt & 145.25  & 126.92  & 111.90  & 178.67  & 333.97  & 287.76 \\ \hline

\end{tabular}

  \end{center}
\end{table}

\bibliographystyle{siam}
\bibliography{ref}

\begin{thebibliography}{10}

\bibitem{a.ss:_lempel_ziv_lz77}
{\sc A.~Al-Hafeedh, M.~Crochemore, L.~Ilie, J.~Kopylov, W.~Smyth, G.~Tischler,
  and M.~Yusufu}, {\em A comparison of index-based {L}empel-{Z}iv {LZ77}
  factorization algorithms}, ACM Computing Surveys,  (in press).

\bibitem{crochemore84:_linear}
{\sc M.~Crochemore}, {\em Linear searching for a square in a word}, Bulletin of
  the European Association of Theoretical Computer Science, 24 (1984),
  pp.~66--72.

\bibitem{crochemore08:_comput_longes_previous_factor}
{\sc M.~Crochemore and L.~Ilie}, {\em Computing longest previous factor in
  linear time and applications}, Information Processing Letters, 106 (2008),
  pp.~75--80.

\bibitem{duval04:_linear}
{\sc J.-P. Duval, R.~Kolpakov, G.~Kucherov, T.~Lecroq, and A.~Lefebvre}, {\em
  Linear-time computation of local periods}, Theoretical Computer Science, 326
  (2004), pp.~229--240.

\bibitem{goto13_lz}
{\sc K.~Goto and H.~Bannai}, {\em Simpler and faster lempel ziv factorization},
  in DCC, 2013, pp.~133--142.

\bibitem{karkkainen_lz}
{\sc J.~K\"{a}rkk\"{a}inen, D.~Kempa, and S.~J. Puglisi}, {\em Linear time
  {L}empel-{Z}iv factorization: Simple, fast, small}, in Proc. CPM'13, 2013.

\bibitem{karkkainen09_plcp}
{\sc J.~K{\"a}rkk{\"a}inen, G.~Manzini, and S.~J. Puglisi}, {\em Permuted
  longest-common-prefix array}, in CPM, 2009, pp.~181--192.

\bibitem{Karkkainen_Sanders_icalp03}
{\sc J.~K\"{a}rkk\"{a}inen and P.~Sanders}, {\em Simple linear work suffix
  array construction}, in Proc. ICALP 2003, 2003, pp.~943--955.

\bibitem{Kasai01}
{\sc T.~Kasai, G.~Lee, H.~Arimura, S.~Arikawa, and K.~Park}, {\em {Linear-time
  Longest-Common-Prefix Computation in Suffix Arrays and Its Applications}}, in
  Proc. CPM 2001, 2001, pp.~181--192.

\bibitem{kolpakov99:_findin_maxim_repet_in_word}
{\sc R.~Kolpakov and G.~Kucherov}, {\em Finding maximal repetitions in a word
  in linear time}, in Proc. FOCS 1999, 1999, pp.~596--604.

\bibitem{kreft11:_self_based_lz77}
{\sc S.~Kreft and G.~Navarro}, {\em Self-indexing based on {LZ77}}, in Proc.
  CPM 2011, vol.~6661 of LNCS, 2011, pp.~41--54.

\bibitem{manber93:_suffix}
{\sc U.~Manber and G.~Myers}, {\em Suffix arrays: A new method for on-line
  string searches}, SIAM J.~Computing, 22 (1993), pp.~935--948.

\bibitem{nong13_sacak}
{\sc G.~Nong}, {\em Practical linear-time {\it o}(1)-workspace suffix sorting
  for constant alphabets}, ACM Trans. Inf. Syst., 31 (2013), p.~15.

\bibitem{nong11_sais}
{\sc G.~Nong, S.~Zhang, and W.~H. Chan}, {\em Two efficient algorithms for
  linear time suffix array construction}, IEEE Trans. Computers, 60 (2011),
  pp.~1471--1484.

\bibitem{rytter03:_applic_lempel_ziv}
{\sc W.~Rytter}, {\em Application of {L}empel-{Z}iv factorization to the
  approximation of grammar-based compression}, Theoretical Computer Science,
  302 (2003), pp.~211--222.

\bibitem{LZ77}
{\sc J.~Ziv and A.~Lempel}, {\em A universal algorithm for sequential data
  compression}, IEEE Transactions on Information Theory, IT-23 (1977),
  pp.~337--343.

\end{thebibliography}

\clearpage
\appendix
\section*{Appendix}
\section{Pseudo-code}
\begin{algorithm2e}[h]
  \caption{
    Computing the LZ77 factorization from $\SA$ via $\PSV$ and $\NSV$ arrays (KKP3)
  }
  \label{algo:pnsv_from_sa}
  \SetKwInOut{Input}{Input}\SetKwInOut{Output}{Output}
  \SetKw{KwAnd}{and}
  \Input{Suffix Array $\SA[1..N]$ of string $T$ of length $N$}
  $\SA[0] \leftarrow 0$ \;
  $\SA[N+1] \leftarrow 0$ \;
  \For{$i \leftarrow 1$ \KwTo $N+1$}{
    \While{$\SA[top] > \SA[i]$}{
      $\NSV[\SA[top]] \leftarrow \SA[i]$ \;
      $\PSV[\SA[top]] \leftarrow \SA[top-1]$ \;
      $top \leftarrow top -1$ \;
    }
    $top \leftarrow top +1$ \;
    $\SA[top] \leftarrow \SA[i]$ \;
  }
  \While{$i \leq n$}{
    $lcp_{nsv} \leftarrow lcp(i, \NSV[i])$ \tcp*[l]{return 0 if $\NSV[i]=0$}
    $lcp_{psv} \leftarrow lcp(i, \PSV[i])$ \tcp*[l]{return 0 if $\PSV[i]=0$}
    $l \leftarrow -1$ \;
    $p \leftarrow T[i]$ \;
    \lIf{$lcp_{nsv} > 0$ \KwAnd $lcp_{nsv} \geq lcp_{psv}$}{
      $l \leftarrow lcp_{nsv}$ ;
      $p \leftarrow \NSV[i]$
    }\lElseIf{$lcp_{psv} > 0$}{
      $l \leftarrow lcp_{psv}$ ;
      $p \leftarrow \PSV[i]$
    }
    \Output{($(l, p)$)}
    $i \leftarrow i + \max (1, l)$ \;
  }
\end{algorithm2e}

\section{Proof of Lemma~\ref{lem:phi2sa}}
\begin{proof}
  The induced sorting algorithm constructs $\SA$ by first computing
  $\LMSSA$ and stores it in $\SA[1..k]$, where $k$ is the number of 
  LMS suffixes.  
  Thus,
  if we can somehow compute $\LMSSA$ from the $\Phia$ array
  in linear time using $O(1)$ extra working space
  and save it in $\SA[1..k]$, we have proved the lemma.

  Let $\Array$ be an integer array of size $N$, used in our algorithm,
  initially equal to the $\Phia$ array.
  Our algorithm will consist of two steps.
  First, for all LMS-suffixes $\suf{i}$, we compute the preceding suffix of
  $\suf{i}$, and store it in $\Array[i]$
  (we store $\Array[i] = \EMPTY$ if $\suf{i}$ is not an LMS suffix), thus
  obtaining an array based linked list representation of LMS-suffixes in
  lexicographically decreasing order.
  Second, we rewrite $\Array$ so that $\Array[1..k] = \LMSSA[1..k]$,
  reversing the procedure described in Step~\ref{phiis:arrange_lms} 
  of Section~\ref{sec:const_phi}.

  For the first step, we compute for each character $c$,
  the starting positions in $\SA$ of the S-interval for $c$
  by counting the number of L-suffixes and S-suffixes that start 
  with the character $c$. As in Step~\ref{sais:putlms} of Section~\ref{subsec:nong},
  this can be done in linear time and constant space.
  We then simulate a right to left traversal on the $\SA$
  using the $\Phia$ array stored as $\Array$.
  Let $\suf{\SA[i]}$ be the suffix that the algorithm currently processing.
  For $\suf{\SA[i]}$ to be an LMS-suffix, it must be that
  $\suf{\SA[i]}$ is an S-suffix, and also $T[\SA[i]-1] > T[\SA[i]]$.  
  The former condition can be checked by 
  whether the position $i$ is in an L-interval or an S-interval.
  During the process, we remember the previous LMS-suffix $\suf{j}$,
  and set $\Array[j]=i$ if $\suf{i}$ is an LMS-suffix,
  and we continue traversing by reading $\Array[i]$ and setting $\Array[i]=\EMPTY$.
  In this way, we can compute a lexicographically decreasing list of
  LMS suffixes, represented in $\Array$ in linear time and $O(1)$ working space.

  Now, we only have to rearrange this list of suffixes to $\LMSSA$.
  The process is the opposite of Step~\ref{phiis:arrange_lms} in Section~\ref{sec:const_phi}.
  We first traverse the LMS suffixes in lexicographically decreasing order.
  We try to set the largest LMS suffix at $\Array[2k]$, 
  the second largest LMS suffix at $\Array[2(k-1)]$ and so on.
  If for the $i$th largest LMS suffix, $\Array[2i] = \EMPTY$, we simply set $\Array[2i]$ to be this value. 
  Otherwise, $2i$ was an LMS-suffix and part of the list.
  In this case, we store the value in $\Array[2i-1]$.
  Notice that again since LMS suffixes cannot start at consecutive positions,
  if $2i$ was an LMS suffix, $2i-1$ cannot be an LMS suffix, and the
  algorithm will never try to set another value at this position.

  Since the original linked list of LMS-suffixes was not overwritten and is preserved, 
  we can traverse this again this time setting the corresponding positions to $\EMPTY$.
  Then, checking all positions $2i$ for $1 \leq i\leq k$, 
  if $\Array[2i] = \EMPTY$ then the corresponding value was stored in $\Array[2i-1]$ and can be retrieved.
  Finally, we copy the values at $\Array[2i]$ to $\Array[i]$ for each $1\leq i\leq k$.
  Thus, $\LMSSA$ can be computed in linear time using $O(1)$ working space.  
\end{proof}

\end{document}